\documentclass{llncs}
\usepackage{natbib} 
\usepackage{pstricks}
\usepackage{pstricks-add}
\usepackage{graphicx}
\usepackage{amsmath}
\usepackage{hyperref}
\usepackage{amssymb}
\usepackage{multirow}
\usepackage{booktabs}
\usepackage{tabularx}
\begin{document}
\renewcommand{\thepage}{\arabic{page}}
\pagestyle{plain} 
\setcounter{page}{1}
\title{A Hybrid Fuzzy Regression Model for Optimal Loss Reserving in Insurance}
\titlerunning{Hamiltonian Mechanics}  
%
\author{WOUNDJIAGU\'{E} Apollinaire\inst{1,4} \and Mbele Bidima Martin Le Doux\inst{2}
\and Waweru Mwangi Ronald\inst{3}}
\authorrunning{Apollinaire WOUNDJIAGU\'{E} et al.} 
%
\tocauthor{WOUNDJIAGU\'{E} Apollinaire, Mbele Bidima Martin Le Doux and Waweru Mwangi Ronald}
\institute{Institute of Basic Sciences Technology and Innovation, Pan African University-Jomo Kenyatta University of Agriculture and Technology, P.O. Box 62000-00200, Nairobi-Kenya,\\
\email{appoappolinaire@yahoo.fr},\\
\texttt{woundjiague.apollinaire@students.jkuat.ac.ke}
\and
Faculty of Sciences, University of Yaounde I, P.O.Box 812 Yaoundé, Cameroon
\and 
School of Computing and Information Technology, Jomo Kenyatta University of Agriculture and Technology, P.O. Box 62000-00200, Nairobi-Kenya
\and National Advanced School of Engineering, University of Maroua, P.O.Box 46 Maroua, Cameroon}
\maketitle   
\begin{abstract}
 In this article, a Hybrid Fuzzy Regression Model with Asymmetric Triangular Fuzzy Coefficients and optimized $h-$value in Generalized Linear Models (GLM) framework have been developed. The weighted functions of Fuzzy Numbers rather than the Expected value of Fuzzy Number is used as a defuzzification procedure. We perform the new model on a numerical data (Taylor and Ashe, 1983) to predict incremental payments in loss reserving. We prove that the new Hybrid Model with the optimized $h-$value produce better results than the classical GLM according to the Reserve Prediction Error and Reserve Standard Deviation.

\end{abstract}
\section{Introduction}
An important role of a non-life actuary is the calculation of provisions, mainly IBNR reserve. Then, finding the fair value of loss reserve is a relevant topic for non-life actuaries. Indeed, insurance companies must simultaneously have enough reserves to meet their commitment to policyholders and have enough funds for their investments. Therefore several methods have been proposed in actuarial science literature to capture this fair value.\par
In one hand, we distinguish deterministic methods (Bornhuetter and Ferguson, 1972; Taylor,
1986; Linnemann, 1984). They provide crisp predictions for reserves. In the other hand, Taylor et al.
(2003); W\"{u}thrich and Merz (2008); Mack (1991); England and Verrall (2002) present stochastic methods. Those methods don't give only a crisp value of the reserves but provide also their variability. But even stochastic methods have weakness.\par
In Straub and Swiss (1988), there are some experiences where stochastic methods can give unrealistic estimates. For example, when the claims are related to body injures, the future losses for the company will depend on the growth of the wage index that help to determine the amount of indemnity, and depends also on changes in court practices and public awareness of liability matters. Then the information is vague. Therefore the use of Fuzzy Set Theory becomes very attractive when the information is vague as in this case. \par 
de Andr\'{e}s S\'{a}nchez (2006); de Andr\'{e}s-S\'{a}nchez (2007, 2012); de Andr\'{e}s S\'{a}nchez (2014) present the interest of fuzzy regression models (FRM) in the calculus of loss reserves in insurance using the concept of expected value of a Fuzzy Number (FN) (de Campos Ib\'{a}\~{n}ez and Mu\~{n}oz, 1989). Asai (1982) is the first to develop a FRM where the coefficient are fuzzy numbers (Dubois and Prade, 1988). In the case of loss reserving, FN are easy to handle arithmetically unlike in the case of classical regression where the coefficients are random variables and are not easy to handle arithmetically. Another difference between fuzzy regression and classical regression is in dealing with errors as fuzzy variables in fuzzy regression modelling while errors are considered as random residuals in classical regression. But to integrate both fuzziness and randomness into a regression model, one should think about hybrid regression models.\par Apaydin and Baser (2010); Baser and Apaydin (2010) proposed a hybrid fuzzy least-squares regression (HFLSR) (Chang, 2001; Apaydin and Baser, 2010; Baser and Apaydin, 2010) analysis in claim reserving framework using a weighted function of fuzzy number (Yager and Filev, 1999).\par
However, the FRM developed in de Andr\'{e}s S\'{a}nchez (2006); de Andr\'{e}s-S\'{a}nchez (2007, 2012);
de Andr\'{e}s S\'{a}nchez (2014) and the HFLSR (Chang, 2001; Apaydin and Baser, 2010; Baser and
Apaydin, 2010) as well don't select a proper value of $h$ and is of the greatest importance. The criteria for selecting an $h$ value are ad hoc (Moskowitz and Kim, 1993).\par
In this paper, we propose a hybrid model with asymmetric triangular fuzzy coefficients (ATFC) based on a GLM (log-Poisson regression) and optimized $h$ value in loss reserving framework using a weighted function of FN.\par
The structure of the paper is as follows. We present in the first section the preliminaries on fuzzy sets and their properties. In the second section, we shall do a review of some models and results on fuzzy regression. In the third section, the framework of estimation of loss reserve with log-Poisson regression will be introduce. We shall propose a new hybrid regression for loss reserving in the fourth section. Then we conclude the article.
\section{Preliminaries on Fuzzy Sets and their Properties}
In this section, we review some concepts related to our research. That is the concept of fuzzy set, membership function, fuzzy number, fuzzy regression and weighted function of FN.
\subsection{Review on some Definitions and Properties of Fuzzy Sets}
\begin{definition}(Zadeh, 1965)\\
Let $\Omega$ be a non empty set and $\omega \in \Omega$. In classical set theory, a subset $A$ of $\Omega$ can be defined by its characteristic function $\chi_{A}$ as a mapping from the elements of $\Omega$ to the elements of the set $\{0, 1\}$ ,
\begin{equation}
\chi_{A} : \Omega \longrightarrow \{ 0, 1\} 
\end{equation}
This mapping may be represented as a set of ordered pairs, with exactly one ordered pair present for each element of $\Omega$. The first element of the ordered pair is an element of the set $\Omega$, and the second element is an element of the set$\{0, 1\}$ . The value zero is used to represent non-membership, and the value one is used to represent membership. The truth or falsity of the statement "$\omega$ is in $A$" is determined by the ordered pair $(\omega, \chi_{A}(\omega) )$. The statement is true if the second element of the ordered pair is 1, and the statement is false if it is 0.\par
Similarly, a fuzzy subset (also called fuzzy set) $\tilde{A}$ of a set $\Omega$ can be defined as a set of ordered pairs, each with the first element from $\Omega$, and the second element from the interval $[0, 1]$, with exactly one ordered pair present for each element of $\Omega$. This defines a mapping called membership function. 
\end{definition}

\begin{definition} (Zadeh, 1965)\\
The membership function of a fuzzy set $\tilde{A}$, denoted by $\mu_{\tilde{A}}$ is defined by
\begin{equation}
\mu_{\tilde{A}} : \Omega \longrightarrow [0,1]
\end{equation}
where $\mu_{\tilde{A}}$ is typically interpreted as the membership degree of element $\omega$ in the fuzzy set $\tilde{A}$.\par
The degree to which the statement " $\omega$ is in $\tilde{A}$" is true is determined by finding the ordered pair $(\omega,\mu_{\tilde{A}}(\omega))$. The degree of truth of the statement is the second element of the ordered pair. A fuzzy set (Zadeh, 1965) $\tilde{A}$ on $\Omega$ can also be defined as a set of tuples:
\begin{equation}
\tilde{A} = \{(\omega,\mu_{\tilde{A}}(\omega)) \mid \omega \in \Omega \}.
\end{equation}
and could be represented by a graphic.
\end{definition}

\begin{definition}(Dubois and Prade, 1978)\\
Let $\Omega$ be the set of objects and $\tilde{A} \subset\Omega.$ The $\alpha-$cut $\tilde{A}_{\alpha}$ of $\tilde{A}$ is the set defined by
\begin{equation*}
\tilde{A}_{\alpha} = \{\omega \in \Omega, \mu_{\tilde{A}}(\omega) \geqslant \alpha \}.  
\end{equation*}	
\end{definition}

\begin{definition} (Dubois and Prade, 1988)\\
\begin{enumerate}
	\item 
A fuzzy number $\tilde{A}$ is a fuzzy set of a universe $\Omega$ (the real line $\mathbb{R}$) such that :
\begin{enumerate}
	\item[a.] 
 all its $\alpha-$cut are convex which is equivalent to $\tilde{A}$ is	convex, that is $\forall \omega_{1}, \omega_{2} \in \mathbb{R}$ and $\lambda \in [0,1], \quad \mu_{\tilde{A}}(\lambda \omega_{1} + (1-\lambda) \omega_{2}) \geqslant \min ( \mu_{\tilde{A}}(\omega_{1}),\mu_{\tilde{A}}(\omega_{2}) )$;
	\item[b.]
 $\tilde{A}$ is	normalized, that is $\exists \omega_{0} \in \Omega$ such that $\mu_{\tilde{A}}(\omega_{0})=1.$
	\item[c.]
 $\mu_{\tilde{A}}$ is continued membership function of bounded support, where $\Omega=\mathbb{R}$ and $[0,1]$ are equipped with the natural topology.	
\end{enumerate}
    \item
A triangular fuzzy number (TFN) $\tilde{\gamma}$ is a fuzzy number denoted by $\tilde{\gamma}=(\beta^{L}, \alpha^{c},\beta^{R}); \quad \beta^{L}, \alpha^{c},\beta^{R}\in \mathbb{R}$, such that $\mu_{\tilde{A}}(\beta^{L})=\mu_{\tilde{A}}(\beta^{R})=0$ and $\mu_{\tilde{A}}(\alpha^{c})=1$ with $\alpha^{c}$ the centre of $\tilde{\gamma}$, $\beta^{L}$ its left spread and $\beta^{R}$ its right spread (Lai and Hwang, 1992).\par
A TFN $\tilde{\gamma}$ could be defined with its membership degree function $\mu_{\tilde{\gamma}}$ or, with its $h-$level ($\alpha-$ cut  ($h\in [0,1]$) $\gamma_{h}$ (see Dubois and Prade (1988)), i.e
\begin{equation}\label{57}
\mu_{\tilde{\gamma}}(x)= \left\{ \begin{array}{cl}
1 - \dfrac{\alpha^{c} - x}{\beta^{L}} & \text{if}\quad \alpha - \beta^{L} < x \leqslant \alpha\\
1 - \dfrac{x - \alpha^{c}}{\beta^{R}} & \text{if}\quad \alpha < x \leqslant \alpha +  \beta^{R}\\
0 & \text{if} \quad \text{otherwise}
\end{array}\right. 
\end{equation}
or
\begin{equation}\label{58}
\tilde{\gamma}_{h}=[\gamma_{L_{h}} , \gamma_{R_{h}}]=\big [\alpha^{c} - \beta^{L}(1-h), \alpha^{c} + \beta^{R} (1-h)  \big ]
\end{equation}
\begin{itemize}
	\item 
	If $\alpha^{c} - \beta^{L} = \beta^{R} - \alpha^{c}$, then $\tilde{\gamma}$ define a STFN
	\item
	Otherwise $\beta^{L} \neq \beta^{R}$, then $\tilde{\gamma}$ define an ATFN (see Figure \ref{fig2}).\\ 
	\begin{figure}[h!]
		\begin{center}
			\caption{Asymmetric triangular fuzzy coefficients $\tilde{\gamma}_{j}=(\beta_{j}^{L},\alpha^{c}_{j},\beta_{j}^{R})$}\label{fig2}
			\bigskip
			\psset{xunit=3cm,yunit=3.25cm}
			\begin{pspicture}(0,0)(3,1)
			\psgrid[gridcolor=lightgray,subgriddiv=2,%
			subgridcolor=lightgray,gridlabels=2mm](0,0)(3,1)
			\psline[linewidth=2pt](0.5,0)(1.5,1)(2,0)
			\psline[linewidth=2pt,linestyle=dashed](0,1)(3,1)
			\psline[linewidth=2pt,linestyle=dashed](1.5,1)(1.5,0)
			\rput{90}(-0.08,0.5){\textbf{membership}}
			\rput{0}(1.9,0.5){$\tilde{\gamma}^{\prime}_{j}$}
			\psline[linewidth=0.75pt](0.5,0)(0.5,-0.2)
			\psline[linewidth=0.75pt](1.5,0)(1.5,-0.2)
			\psline[linewidth=0.75pt](2,0)(2,-0.2)
			\pcline[offset=-12pt]{<->}(0.5,-0.05)(1.5,-0.05)
			\pcline[offset=-12pt]{<->}(1.5,-0.05)(2,-0.05)
			\uput[-90](0.5,-0.2){$\beta_{j}^{L}$}
			\uput[45](1.5,-0.2){$\alpha^{c}_{j}$}
			\uput[-90](2,-0.2){$\beta_{j}^{R}$}
			\end{pspicture}
		\end{center}\end{figure}
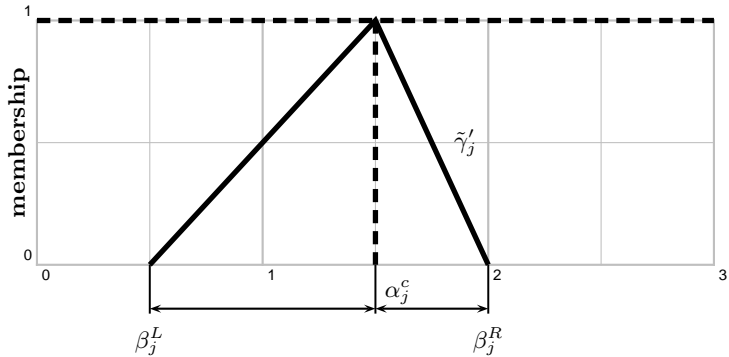
		\bigskip
\end{itemize}
\end{enumerate}
\end{definition}

\paragraph{Notes and Comments.}
It is well know that if $\tilde{A}$ is a fuzzy number, then $\tilde{A}_{h}$, the $h$ level ($\alpha-$cut) of $\tilde{A}$ is a compact set of $\mathbb{R}$, for all $h \in [0,1].$

\begin{property} (de Andr\'{e}s S\`{a}nchez, 2006)
Let $f(\tilde{\Gamma})=f(\tilde{\gamma}_{1},\tilde{\gamma}_{2}, \ldots , \tilde{\gamma}_{k})$ be a $k-$vector of TFN such that $\tilde{\gamma}_{i}=(\beta_{i}^{L}, \alpha^{c}_{i}, \beta_{i}^{R})  \in \mathbb{R}^3, i=1,\ldots , k$ are triangular fuzzy numbers (TFN). 
\begin{enumerate}
	\item[1.]
	If $f(\tilde{\Gamma})$ is obtained from a linear combination of the TFN $\tilde{\gamma}_{i}, i=1,\ldots , k$, then $f(\tilde{\Gamma})=(f^{L}(\Gamma), f^{c}(\Gamma), f^{R}(\Gamma))$ is also a TFN; where
	\begin{align}
	f^{L}(\Gamma) & =\sum\limits_{i=1\atop r_{i}\geqslant 0}^{k}\beta_{i}^{L}|r_{i}| + \sum\limits_{i=1\atop r_{i}< 0}^{k}\beta_{i}^{R}|r_{i}|\\
	f^{c}(\Gamma) & = \sum\limits_{i=1}^{k}\alpha^{c}_{i}r_{i}\\
	f^{R}(\Gamma) & = \sum\limits_{i=1\atop r_{i}\geqslant 0}^{k}\beta_{i}^{R}|r_{i}| + \sum\limits_{i=1\atop r_{i}< 0}^{N}\beta_{i}^{L}|r_{i}|
	\end{align}
	From the extension principle Zadeh (975b,c,d), we can obtain the $h-$ level of $f(\tilde{\Gamma})$, i.e
	\begin{equation}
	[f(\Gamma)]_{h} = f(\gamma_{1_{h}}, \gamma_{2_{h}}, \ldots , \gamma_{k_{h}})
	\end{equation}
	and 
	\begin{align}
	f(\tilde{\Gamma}) & =\sum\limits_{i=1}^{k}r_{i}\tilde{\gamma}_{i}\\
	& = \big (f^{L}(\Gamma), f^{c}(\Gamma), f^{R}(\Gamma) \big ), \quad \text{with} \quad r_{i} \in \mathbb{R}          
	\end{align}
	\item[2.]
	If $f(\tilde{\Gamma})$ is evaluated by non-linear functions with TFN, i.e $f(\cdot)$  is
	increasing with respect to the first $n$ variables, where $n \leqslant k$, and decreasing with respect to $k-n$ variables, the result will not be a TFN. But (Dubois and Prade, 1988) has shown that $f(\tilde{\Gamma})$ can be approximate with a TFN $f^{\prime}(\tilde{\Gamma})=\big (f^{L^{\prime}}(\Gamma), f^{c^{\prime}}(\Gamma), f^{R^{\prime}}(\Gamma)   \big )$. i.e
	\begin{align}
	f^{L^{\prime}}(\Gamma) & =\sum\limits_{i=1}^{n} \dfrac{\partial f^{c}(\Gamma)}{\partial \alpha^{c}_{i}}\beta_{i}^{L} - \sum\limits_{i=n+1}^{k} \dfrac{\partial f^{c}(\Gamma)}{\partial \alpha^{c}_{i}}\beta_{i}^{R}\\
	f^{c^{\prime}}(\Gamma) & = f(\alpha^{c}_{1},\alpha^{c}_{2}, \ldots , \alpha^{c}_{k})=f^{c}(\Gamma)\\
	f^{R^{\prime}}(\Gamma) & =\sum\limits_{i=1}^{n} \dfrac{\partial f^{c}(\Gamma)}{\partial \alpha^{c}_{i}}\beta_{i}^{R} - \sum\limits_{i=n+1}^{k} \dfrac{\partial f^{c}(\Gamma)}{\partial \alpha^{c}_{i}}\beta_{i}^{L}
	\end{align}
\end{enumerate}	
\end{property}
\begin{definition} (weighted function of FN)\\
	Let $\tilde{\gamma}=(\beta^{L}, \alpha^{c},\beta^{R}),\quad \beta^{L}, \alpha^{c},\beta^{R}\in \mathbb{R}$ be an asymmetric normal triangular FN and $g$ be a regular or nonregular function of $\tilde{\gamma}$ (Kauffman and Gupta, 1991). The weighted function of FN associated with the valuation method of Yager and Filev (1999) is defined as :
	\begin{equation}\label{56}
	g(\tilde{\gamma})=\dfrac{\frac{1}{2} \big [\int_{0}^{1}f(\gamma_{L_{h}})hdh  + \int_{0}^{1}f(\gamma_{R_{h}})hdh\big ]}{\int_{0}^{1}hdh}
	\end{equation}
\end{definition}\label{def1}
\section{Review of some Models and Results on Fuzzy Regression}
In this section we review the FRM proposed by Asai (1982) and the one proposed by Ishibuchi and
Nii (2001) and we present their properties. Those models will allow us to develop the new model in loss reserving. 
\subsection{Ishibuchi's FRM with asymmetric triangular fuzzy coefficients}\label{22}
Let us define the fuzzy linear regression (FLR) model proposed by Asai (1982).\par
Let $(Y_{N\times 1}, \mathbf{X}_{N\times (m+1)})$ the given crisp data and let
\begin{align}\label{8}
\tilde{Y}_{N\times 1} & =f(\mathbf{X}_{N\times (m+1)}) = \mathbf{X}_{N\times (m+1)} \otimes \tilde{\Gamma}_{(m+1)\times 1}\\ \nonumber
\begin{pmatrix}
\tilde{Y}_{1}\\
\vdots \\
\tilde{Y}_{N}
\end{pmatrix}
& = \begin{pmatrix}
\mathbf{x}_{1} \\
\vdots\\
\mathbf{x}_{N}
\end{pmatrix}
\otimes
\begin{pmatrix}
\tilde{\gamma}_{0}\\
\vdots\\
\tilde{\gamma}_{m}
\end{pmatrix}\\
& =\begin{pmatrix}
\mathbf{x}_{1}\\
\vdots \\
\mathbf{x}_{N}
\end{pmatrix}
\otimes
\begin{pmatrix}
(\alpha^{c}_{0}, \beta_{0})\\
\vdots \\
(\alpha^{c}_{m}, \beta_{m})
\end{pmatrix}\\
& =
\left(\begin{array}{c}
\big(\sum\limits_{j=0}^{m}\alpha^{c}_{j}x_{1j}, \sum\limits_{j=0}^{m}\beta_{j}x_{1j} \big)\\
\vdots\\
\big(\sum\limits_{j=0}^{m}\alpha^{c}_{j}x_{Nj}, \sum\limits_{j=0}^{m}\beta_{j}x_{Nj} \big)
\end{array} \right)
\end{align}  
be a FLR model with symmetric triangular fuzzy coefficients (STFC), where $\tilde{Y}$ is a fuzzy output from $f(\mathbf{X}_{N\times (m+1)})$ and $\mathbf{X}_{N\times (m+1)}$ is the matrix of the given crisp dataset, $\tilde{\Gamma}_{(m+1)\times 1} = \bigg[ (\alpha^{c}_{0}, \beta_{0}), (\alpha^{c}_{1}, \beta_{1}), \ldots , (\alpha^{c}_{m}, \beta_{m})\bigg]^{T}$ is the fuzzy parameter of the model.\par
In model \eqref{8}, $\tilde{\gamma}_{j} = (\alpha^{c}_{j}, \beta_{j})\in \mathbb{R}^2, j=0,\ldots,m$ are fuzzy coefficients such that $\alpha^{c}_{j}$ are centres of $\tilde{\gamma}_{j}$, and $\beta_{j}$ are its spread. The disturbance term is not introduced as a random addend in the linear relation, but incorporated into the coefficients $\tilde{\gamma}_{j}, j=0,\ldots , m$.\par
When the coefficients are symmetric triangular fuzzy number (STFN), the output $\tilde{Y}$ is also a STFN.\par
Let us denote
\begin{align}
f(\mathbf{X}_{i}) & =( f^{\alpha^{c}}(\mathbf{X}_{i}), f^{\beta}(\mathbf{X}_{i})),\quad i=1, \ldots , N\\
f(\mathbf{X}_{N\times (m+1)}) & =\bigg[f(\mathbf{X}_{1}),\ldots ,f(\mathbf{X}_{N})  \bigg]^{T}
\end{align}
where
\begin{align}
f^{\alpha^{c}}(\mathbf{X}_{i}) & =\alpha^{c}_{0} + \sum\limits_{j=1}^{m}\alpha^{c}_{j}\cdot x_{ij}\\
f^{\beta}(\mathbf{X}_{i}) & = \beta_{0} + \sum\limits_{j=1}^{m} \beta_{j}\cdot |x_{ij}|,\quad i=1, \ldots , N 
\end{align}
and $\mathbf{X}_{i}=(x_{i1}, x_{i2}, \ldots , x_{im})$.\\
Then $f(\mathbf{X}_{i})$ is a STFN. Its $h-$level, i.e $\alpha -$cut with $\alpha = h$ is calculated as follow for $h \in [0,1]$.
\begin{equation}\label{10}
[f(\mathbf{X}_{i})]_{h} = \big[f^{\alpha^{c}}(\mathbf{X}_{i}) - (1-h)\cdot f^{\beta}(\mathbf{X}_{i}), \quad f^{\alpha^{c}}(\mathbf{X}_{i}) + (1-h)\cdot f^{\beta}(\mathbf{X}_{i}) \big],\quad i=1, \ldots , N 
\end{equation} 
From \eqref{10}, $Y_{i} \in [f(\mathbf{X}_{i})]_{h}, \quad i=1,\ldots , N\Rightarrow $
\begin{equation}\label{11}
\left\{ \begin{array}{cl}
f^{\alpha^{c}}(\mathbf{X}_{i}) - (1-h)\cdot f^{\beta}(\mathbf{X}_{i})\leqslant Y_{i}\\
f^{\alpha^{c}}(\mathbf{X}_{i}) + (1-h)\cdot f^{\beta}(\mathbf{X}_{i})\geqslant Y_{i}, i=1,2,\ldots,N 
\end{array}\right.
\end{equation}
To determine the parameters $\tilde{\gamma}_{j}=(\alpha^{c}_{j}, \beta_{j})$ of the FRM \eqref{8}, Asai (1982) propose to solve a linear programming problem with an objective function of minimizing the total spread of the fuzzy coefficient, i.e
\begin{equation}\label{40}
\left\{ \begin{array}{cl}
\min: & z=\sum\limits_{i=1}^{N}f^{\beta}(\mathbf{X}_{i})=\sum\limits_{i=1}^{N}\big \{\beta_{0} + \beta_{1}\cdot |x_{i1}| + \ldots + \beta_{m}\cdot |x_{im}| \big\},\\
\text{subject to} & \left\{ \begin{array}{cl} f^{\alpha^{c}}(\mathbf{X}_{i}) - (1-h)\cdot f^{\beta}(\mathbf{X}_{i})\leqslant Y_{i}\\
f^{\alpha^{c}}(\mathbf{X}_{i}) + (1-h)\cdot f^{\beta}(\mathbf{X}_{i})\geqslant Y_{i},\\
\beta_{j}\geqslant 0 , i=1,\ldots , N; j=0,1,\ldots , m
\end{array}\right.
\end{array}\right. 
\end{equation}

Ishibuchi and Nii (2001) have shown that the fuzzy regression (FR) method developed by Asai (1982), when applied to different data sets can provide the same linear fuzzy model by solving the linear programming problem in \eqref{40} for $h=0.5$. Then they proposed the FRM with ATFC in order to remedy this limitation.\par
Let us assume now that
\begin{equation}\label{14}
\tilde{\gamma}_{j}=(\beta_{j}^{L}, \alpha^{c}_{j}, \beta_{j}^{R}), j=0,1,\ldots , m
\end{equation}
be ATFC in the fuzzy model \eqref{8}, where $\beta_{j}^{L}$ is its left spread, $\alpha^{c}_{j}$ its centre and $\beta_{j}^{R}$ its right spread (see Figure \ref{fig2}).\\ 
When $\forall j=0,1,\ldots , m, \quad \tilde{\gamma}_{j}$ are ATFC in model \eqref{8}, then $f(\mathbf{X}_{N\times (m+1)})$ is also calculated as an ATFC (Ishibuchi and Nii, 2001). Hence, we denote
	\begin{equation}
	f(\mathbf{X}_{N\times (m+1)})=\begin{pmatrix}
	f(\mathbf{X}_{1}) \\
	\vdots\\
	f(\mathbf{X}_{N})
	\end{pmatrix}
	=
	\begin{pmatrix}
	\big( f^{\beta^{L}}(\mathbf{X}_{1}), f^{\alpha^{c}}(\mathbf{X}_{1}), f^{\beta^{R}}(\mathbf{X}_{1})\big) \\
	\vdots\\
	\big( f^{\beta^{L}}(\mathbf{X}_{N}), f^{\alpha^{c}}(\mathbf{X}_{N}), f^{\beta^{R}}(\mathbf{X}_{N})\big)
	\end{pmatrix}
	\end{equation}   
	From Kaufmann and Gupta (1991), $f^{\beta^{L}}(\mathbf{X}_{i}),f^{\alpha^{c}}(\mathbf{X}_{i})$ and $f^{\beta^{R}}(\mathbf{X}_{i}), i=1,\ldots , N$ are compute as following
	\begin{align}\label{15}
	f^{\beta^{L}}(\mathbf{X}_{i}) & = \beta^{L}_{0} + \sum\limits_{j=1\atop x_{ij} \geqslant 0}^{m}\beta^{L}_{j}\cdot x_{ij} + \sum\limits_{j=1\atop x_{ij} < 0}^{m}\beta^{R}_{j}\cdot x_{ij}\\\label{151}
	f^{\alpha^{c}}(\mathbf{X}_{i}) & = \alpha^{c}_{0} + \sum\limits_{j=1}^{m}\alpha^{c}_{j}\cdot x_{ij}\\ \label{152}
	f^{\beta^{R}}(\mathbf{X}_{i}) & = \beta^{R}_{0} + \sum\limits_{j=1\atop x_{ij} \geqslant 0}^{m}\beta^{R}_{j}\cdot x_{ij} + \sum\limits_{j=1\atop x_{ij} < 0}^{m}\beta^{L}_{j}\cdot x_{ij}
	\end{align}
	and the $h-$level of $f(\mathbf{X}_{i})$ is as following
	\begin{equation}
	\big[f(\mathbf{X}_{i}) \big]_{h} = \bigg[h\cdot f^{\alpha^{c}}(\mathbf{X}_{i}) - (1-h)\cdot f^{\beta^{L}}(\mathbf{X}_{i}), \quad h\cdot f^{\alpha^{c}}(\mathbf{X}_{i}) + (1-h)\cdot f^{\beta^{R}}(\mathbf{X}_{i})\bigg] \quad i=1,\ldots , N
	\end{equation}
	The steps to determine the fuzzy coefficients in \eqref{14} are as following (Ishibuchi and Nii, 2001) :
	\begin{enumerate}
		\item[$\bullet$]
		Determine $f^{\alpha^{c}}(\mathbf{X}_{i}) \quad i=1,\ldots , N$ by Ordinary Least Square (OLS) regression,
		\item[$\bullet$]
		Determine $f^{\beta^{L}}(\mathbf{X}_{i}), f^{\beta^{R}}(\mathbf{X}_{i})$ of $f(\mathbf{X}_{i})$ for $i=1,\ldots , N$ by solving the linear programming problem :
		\begin{equation}\label{16}
		\left\{ \begin{array}{cl}
		\min: & z=\sum\limits_{i=1}^{N} \bigg \{f^{\beta^{R}}(\mathbf{X}_{i}) - f^{\beta^{L}}(\mathbf{X}_{i})\bigg\},\\
		\text{subject to} & \left\{ \begin{array}{cl} &  h\cdot f^{\alpha^{c}}(\mathbf{X}_{i}) - (1-h)\cdot f^{\beta^{L}}(\mathbf{X}_{i})\leqslant Y_{i}\\
		& h\cdot f^{\alpha^{c}}(\mathbf{X}_{i}) + (1-h)\cdot f^{\beta^{R}}(\mathbf{X}_{i})\geqslant Y_{i},\; i=1,2,\ldots,N\\
		& \beta^{L}_{j}\leq \alpha^{c}_{j} \leq \beta^{R}_{j}, j=0,1,\ldots , m
		\end{array}\right.
		\end{array}\right.
		\end{equation}
		\begin{equation}
		\eqref{15}-\eqref{151}-\eqref{152}-\eqref{16} \Rightarrow \left\{ \begin{array}{cl}
		\min : & z=\sum\limits_{i=1}^{N}\bigg \{\sum\limits_{j=0\atop x_{ij} \geqslant 0}^{m}(\beta_{j}^{R}-\beta_{j}^{L} )x_{ij} + \sum\limits_{j=1\atop x_{ij}< 0}^{m}(\beta_{j}^{L}-\beta_{j}^{R} )x_{ij}\bigg\}\\
		\text{s/t} & \left\{ \begin{array}{cl}
		& h\cdot\big(\alpha_{0}^{c} + \sum\limits_{j=1}^{m} \alpha_{j}^{c} x_{ij}\big) - (1-h)\big( \beta^{L}_{0} + \sum\limits_{j=1\atop x_{ij} \geqslant 0}^{m}\beta^{L}_{j}\cdot x_{ij} + \sum\limits_{j=1\atop x_{ij} < 0}^{m}\beta^{R}_{j}\cdot x_{ij} \big) \leqslant Y_{i}\\
		& h\cdot\big(\alpha_{0}^{c} + \sum\limits_{j=1}^{m} \alpha_{j}^{c} x_{ij}\big) + (1-h)\big(\beta^{R}_{0} + \sum\limits_{j=1\atop x_{ij} \geqslant 0}^{m}\beta^{R}_{j}\cdot x_{ij} + \sum\limits_{j=1\atop x_{ij} < 0}^{m}\beta^{L}_{j}\cdot x_{ij} \big) \geqslant Y_{i},\\ & i=1,\ldots , N\\
		& \beta_{j}^{L} \leqslant \alpha_{j}^{c}\leqslant \beta_{j}^{R}, \; j=0,1,\ldots, m
		\end{array}\right.
		\end{array}\right.
		\end{equation}\label{153}
	\end{enumerate}
\subsection{Optimizing the $h$ value for FLR with ATFC}
Let us present in this sub section the optimising $h$ value for ATFC developed by Chen et al. (2016) that we shall use later.
\subsubsection{Preliminaries and some Definitions}
Let us consider the model \eqref{8}, but by considering
\begin{align}\label{41}
\tilde{\Gamma}_{(m+1)\times 1} & = \bigg[\tilde{\gamma}_{0}, \tilde{\gamma}_{1}, \ldots , \tilde{\gamma}_{m} \bigg]^{T}\\
& = \bigg[(\beta_{0}^{L}, \alpha_{0}^{c}, \beta_{0}^{R}), (\beta_{1}^{L}, \alpha_{1}^{c}, \beta_{1}^{R}), \ldots , (\beta_{m}^{L}, \alpha_{m}^{c}, \beta_{m}^{R}) \bigg]^{T}.
\end{align}
That is a FLR with ATFC.\par
The system fuzziness in this case (model \eqref{8} with \eqref{41}) is defined by
\begin{align}\label{42}
\nabla & = \sum\limits_{i=1}^{N}\nabla \tilde{Y}_{i}\\
       & = \dfrac{1}{2}\sum\limits_{i=1}^{N} \big( f^{\beta^{R}}(\mathbf{x}_{i}) - f^{\beta^{L}}(\mathbf{x}_{i})\big)\\
       & = \dfrac{1}{2}\sum\limits_{i=1}^{N}\bigg(\sum\limits_{j=1\atop x_{ij}\geqslant 0}^{m}\big(\beta_{j}^{L} + \beta_{j}^{R}\big)x_{ij} -\sum\limits_{j=1\atop x_{ij}< 0}^{m}\big(\beta_{j}^{L} + \beta_{j}^{R}\big)x_{ij} + \big( \beta_{0}^{R} - \beta_{0}^{L}\big)\bigg)
\end{align}
Then the area where $Y_{i}$ is predicted is exactly the fuzziness $\nabla \tilde{Y}_{i}$. That why the objective function in \eqref{16} is to minimize the total fuzziness.\par
\begin{definition}(Chen et al., 2016)\\
The credibility of $\tilde{Y}_{i}$ in representing $Y_{i}$ denoted by $Cr_{i}$ is defined as
\begin{equation}\label{43}
Cr_{i} = \dfrac{\mu_{\tilde{Y}_{i}}(Y_{i})}{\nabla \tilde{Y}_{i}}.
\end{equation}
and the system credibility in model \eqref{8}-\eqref{41}, denoted by $Cr$ is calculated as follow :
\begin{align}
Cr & =\sum\limits_{i=1}^{N} Cr_{i}\\
& = \sum\limits_{i=1}^{N} \dfrac{\mu_{\tilde{Y}_{i}}(Y_{i})}{\nabla \tilde{Y}_{i}}.
\end{align}
The higher the $Cr_{i}$ (resp. $Cr$) is the better the performance of $\tilde{Y}_{i}$ (resp. FLR) will be.
\end{definition}
\begin{definition}(Chen et al., 2016)\\
\begin{enumerate}
	\item define $S_{i}^{h^l}$ by
\begin{equation}\label{46}
S_{i}^{h^l}=
\left\{ \begin{array}{cl}
\dfrac{\hat{Y}_{i}^{c}-Y_{i}}{(Y_{i}^{L})^{h^{l\ast}}}  = \dfrac{\sum\limits_{j=0}^{m} \hat{\alpha}_{j}^{c}x_{ij} - Y_{i}}{(\beta_{0}^{L})^{h^{l\ast}} + \sum\limits_{j=1\atop x_{ij}\geqslant 0}^{m} (\beta_{j}^{L})^{h^{l\ast}}x_{ij} + \sum\limits_{j=1\atop x_{ij}< 0}^{m} (\beta_{j}^{R})^{h^{l\ast}}x_{ij}} & \text{if}\quad Y_{i} \leqslant \hat{Y}_{i}^{c}\\
\dfrac{Y_{i} - \hat{Y}_{i}^{c}}{(Y_{i}^{R})^{h^{l\ast}}}  = \dfrac{Y_{i} - \sum\limits_{j=0}^{m} \hat{\alpha}_{j}^{c}x_{ij}}{(\beta_{0}^{R})^{h^{l\ast}} + \sum\limits_{j=1\atop x_{ij}\geqslant 0}^{m} (\beta_{j}^{R})^{h^{l\ast}}x_{ij} + \sum\limits_{j=1\atop x_{ij}< 0}^{m} (\beta_{j}^{L})^{h^{l\ast}}x_{ij}} & \text{if}\quad Y_{i} > \hat{Y}_{i}^{c}	\\
i=1,\ldots , N; \quad l=1,2
\end{array}\right.
\end{equation}
\item Define $Cr_{i}^{h^{l}}$ by
\begin{align}\label{47}
Cr_{i}^{h^{l}} & = \dfrac{(\mu_{\tilde{Y}_{i}}(Y_{i}) )^{h^{l\ast}}}{\nabla \tilde{Y}_{i}^{h^{l\ast}}}\\
& = \dfrac{2(1-S_{i}^{h^{l\ast}})}{\sum\limits_{j=0}^{m} \big((\beta_{j}^{L})^{h^{l\ast}} + (\beta_{j}^{R})^{h^{l\ast}} \big) |x_{ij}|}, \quad i=1, \ldots , N; \quad l=1,2
\end{align} 
\item Define $p_{i}^{h^1}$ by
\begin{equation}\label{52}
p_{i}^{h^l} = \dfrac{S_{i}^{h^l}}{\nabla \tilde{Y}_{i}^{h^{l\ast}}}, \quad i=1,\ldots , N; \quad l=1,2.
\end{equation}
\item Define $p^{h^l}$ by
\begin{equation*}
p^{h^l}= \sum\limits_{i=1}^{N}p_{i}^{h^l}, \quad l=1,2.
\end{equation*}
\end{enumerate}
\end{definition}
\subsubsection{Optimizing the $h$ value}
\begin{theorem}
Consider $h^1 = 0$ and $h^2 = h$. Then, the optimal value of $h$ is given by
	\begin{equation}\label{55}
	h^{\ast}=
	\left\{ \begin{array}{cl}
	\dfrac{1}{2}\big( 1 - \dfrac{Cr^{0}}{p^{0}}\big) & \text{if} \quad 0 \leqslant \dfrac{Cr^{0}}{p^{0}} \leqslant 1\\
	0 & \text{if} \quad \dfrac{Cr^{0}}{p^{0}}> 1
	\end{array}\right.
	\end{equation}
\end{theorem}
\begin{proof}
(see Chen et al. (2016)).
\end{proof}
\section{Estimation of Loss Reserve with Log-Poisson Regression}
In this section, we are interested in the estimation of loss reserve using a Generalized Linear Model (GLM), especially the log-Poisson regression. We consider a non-life insurance company which sells policies in a period of time (year). This year is referred to as underwriting year. The claims regarding an underwriting year will not necessarily all be paid within this year. Due to legal issues, general consideration of the claim, the delay from the claim's occurrence time to the reporting time, \ldots, some claims are reported and paid in the following years.\par
At some point in time there will however be no more payments regarding underwriting year one; the year one is said to has \textit{run off}. We use then the historical data of claims presented as a \textit{run-off triangle} (\ref{tab1}), where $Y_{ij}$ is the total loss regarding the underwriting period $i$ which have been paid with $j$ periods delay. The loss amounts $Y_{ij}$ with $i+j=k$ have been paid in calendar year $k\in \mathbb{N}$. At period $s\in \mathbb{R}$, we have observed the payments
\begin{equation}\label{59}
Y_{ij}, \quad (i,j)\in \mathcal{T}_{s}
\end{equation}
where
\begin{equation}\label{60}
\mathcal{T}_{s}=\{(i,j)\in \mathbb{N}^{\ast}\times \mathbb{N}: 1 \leqslant i+j \leqslant s \}
\end{equation}
Table \ref{tab1} is usually called \textit{run-off triangle} for example in $s_{0}\in \mathbb{N}$ periods because $\forall j> s_{0}$, $Y_{ij}=0$. And then the reserve $R_{i}$ for the underwriting period $i$ is defined as the predictor of the not yet observed amount $Y_{i,k-1} + \ldots , Y_{i, s_{0}}$. The total amount reserve $R$ is defined as the prediction of 
\begin{equation}\label{61}
\sum\limits_{i,j}Y_{i,j}, \quad (i,j)\in \mathcal{Q}_{k}
\end{equation}  
where
\begin{equation}\label{62}
\mathcal{Q}_{k} = \{(i,j) \in \{1,\ldots, k\}\times \{0,1,\ldots,k\}: \quad i+j \geqslant k+1 \}.
\end{equation}
\begin{table}[!h]
	\centering
	\begin{tabular}{lclclclclclc}
		\toprule
		\multicolumn{8}{c}{\bf Development Year}\\ 
		& & \bf 0 & \bf 1 & \bf $\ldots$ & $l$ & $\ldots$ & $k-i$ & $\ldots$ & $k-1$ & $k$ \\ \hline
		\multirow{9}{6mm}{\rotatebox{90}{\bf Accident Year}} & \bf 0  & $Y_{0,0}$ & $Y_{0,1}$ & $\ldots$ & $Y_{0,l}$ & $\ldots$ & $Y_{0,k-i}$ & $\ldots$ & $Y_{0,k-1}$ & $Y_{0,k}$\\
		& \bf 1   & $Y_{1,0}$ & $Y_{1,1}$ & $\ldots$ & $Y_{1,l}$ & $\ldots$ & $Y_{1,k-i}$ & $\ldots$ & $Y_{1,k-1}$ \\
		& $\vdots$  & $\vdots$ & $\vdots$ & $\ldots$ & $\vdots$ & $\ldots$ & $\vdots$ \\
		& $i$   & $Y_{i,0}$ & $Y_{i,1}$ & $\ldots$ & $Y_{i,l}$ & $\ldots$ & $Y_{i,k-i}$\\
		& $\vdots$  & $\vdots$ & $\vdots$ & $\ldots$ & $\vdots$ \\
		& $k-l$  & $Y_{k-l,0}$ & $Y_{k-l,1}$ & $\ldots$ & $Y_{i,k-i}$\\
		& $\vdots$  & $\vdots$ & $\vdots$\\
		& $k-1$  & $Y_{k-1,0}$ & $Y_{k-1,1}$\\ 
		& $k$  & $Y_{k,0}$\\ \hline
	\end{tabular}
	\caption{Triangle of observed incremental payments}\label{tab1}
\end{table}
Here we assume that $Y_{ij}$ follow a Poisson distribution with the underwriting period $i$ which is reported with $j$ period delay.\\
Assume that $Y_{ij}$, $(i,j)\in \mathcal{T}_{k} \cup \mathcal{Q}_{k}$ are mutually independent (Renshaw and Verrall, 1998) and 
\begin{align}\label{63}
& Y_{ij} \sim \mathcal{P}(e^{\nu_{ij}}),\quad \text{with} \quad \mathbb{V}(Y_{ij})=\psi \mathbb{E}(Y_{ij})\\
& \text{i.e}\; \ln \mathbb{E}(Y_{ij}) = \nu_{ij}\nonumber
\end{align}
where
\begin{equation}\label{64}
\nu_{ij}= \tau + \alpha_{i} + \beta_{j}, \quad (i,j) \in \mathcal{T}_{k} \cup \mathcal{Q}_{k}, \quad \alpha_{1}=\beta_{0}=0,
\end{equation}
$\psi$ is the dispersion parameter, $\tau$ means that we assume the portfolio to grow, or shrink, by a fixed percentage each year, $\alpha_{i}$ means that the proportion settled decreases by a fixed fraction with each origin year and $\beta_{j}$ means that the proportion settled decreases by a fixed fraction with each development year.\\
The parameter vector is given by
\begin{equation}\label{65}
\theta_{k}=(\alpha_{1},\ldots,\alpha_{k}, \beta_{0}, \ldots, \beta_{k-1}, \tau )\in \mathbb{R}^{2k-1}
\end{equation}
The estimator of $Y_{ij}$ can be given using the following relationship
\begin{equation}
\hat{Y}_{ij}=e^{ \hat{ \tau} + \hat{\alpha}_{i} + \hat{\beta}_{j}}
\end{equation}
where $\hat{\tau}, \hat{\alpha}_{i}, \hat{\beta}_{j}$ are the Maximum Likelihood Estimators (MLE) of $\tau, \alpha_{i}, \beta_{j}$ respectively and could be derive from a recursive algorithm (see (Mack, 1991)).\par
Let us consider the numerical example in Table \ref{tab2} (Taylor and Ashe, 1983). This dataset is used in many texts on loss reserving problems.\\
\begin{table}[!h]
	\centering
	\begin{tabular}{|c|c|cccccccccc|}
	\cline{3-12}
	\multicolumn{2}{c|}{} &	\multicolumn{10}{|c|}{\centering Development Year}\\ \hline
	\multicolumn{2}{|c|}{$i/j$} &  1 &  2 &  3 &  4 &  5 &  6 & 7 & 8 & 9 & 10  \\ \hline
		\multirow{9}{6mm}{\rotatebox{90}{\bf Origin Year}} & 1  & 357848 & 766940 & 610542 & 482940 & 527326 & 574398 & 146342 & 139950 & 227229 & 67948 \\
		&  2  & 352118 & 884021 & 933894 & 1183289 & 445745 & 320996 & 527804 & 266172 & 425046 &\\
		&  3  & 290507 & 1001799 & 926219 & 1016654 & 750816 & 146923 & 495992 & 280405 &  &  \\
		&  4  & 310608 & 1108250 & 776189 & 1562400 & 272482 & 352053 & 206286 & & & \\
		&  5  & 443160 & 693190 & 991983 & 769488 & 504851 & 470639 & & & & \\
		&  6  & 396132 & 937085 & 847498 & 805037 & 705960 & & & & & \\
		&  7  & 440832 & 847631 & 1131398 & 1063269 & & & & & &\\
		&  8  & 359480 & 1061648 & 1443370 &  & & & & & &\\ 
		&  9  & 376686 & 986608 &  &  & & & & & &\\ 
		&  10 & 344014 &  &  &  & & & & & &\\ \hline
	\end{tabular}
	\caption{Numerical example from Taylor and Ashe (1983)}\label{tab2}
\end{table}
After performed the log-Poisson regression, we got the following Maximum Likelihood Estimations of the parameters of the model and their confidence interval in Table \ref{tab3}.\\
\begin{table}[!h]
	\centering
	\begin{tabular}{|c|c|c|c|c|}
		\hline
		& $(\hat{\alpha}_{i})_{2 \leqslant i \leqslant 10}$ & $(\hat{\beta}_{j})_{2 \leqslant j \leqslant 10}$  & $CI_{99\%}(\alpha_{i}),\; 2 \leqslant i \leqslant 10$ & $ CI_{99\%}(\beta_{j}), \; 2 \leqslant j \leqslant 10$  \\ \hline \hline
		& 0.331  & 0.912 & [ 0.329, 0.333] & [0.910,0.914] \\
		& 0.321 & 0.958 & [0.318,0.323] & [0.956,0.961]  \\
		& 0.305 & 1.025 & [0.303,0.308]  & [1.023,1.028]\\
		& 0.219  & 0.435 & [0.216,0.221]  & [0.432, 0.437]  \\
		& 0.270  & 0.080 & [0.267,0.272]  & [0.077,0.083]\\
		& 0.372  &-0.006 & [0.369,0.375]  & [-0.0098,-0.0029]\\ 
		& 0.553  &-0.394 & [0.550,0.556]  & [-0.398,-0.390]\\ 
		& 0.368  & 0.009 & [0.365,0.372]  & [0.0048,0.0139]\\ 
		& 0.242  &-1.379 & [0.236,0.248]  & [-1.393,-1.367]\\ \hline
		& \multicolumn{2}{c|}{ $\hat{\tau}=12.506 $} & \multicolumn{2}{c|}{ $CI_{99\%}(\tau)=[12.503,12.508] $}\\ \hline
	\end{tabular}
	\caption{MLE of parameters}\label{tab3}
\end{table}
With a threshold of $1\%$, we conclude that all the coefficients of the regression are significant. As an example of interpretation of the results, the estimation of the payment of $2^{nd}$ origin year could be $e^{0.331 + 0.912 + 12.506} = 936779.4925$.\par
Now let us test if the model performed is adapted to a statistical perspective. Recall that we assumed equi-dispersion in the dataset, i.e $\psi = 1.$ \\
The dispersion test is displayed in Table \ref{tab4}.\\
\begin{table}[!h]
	\centering
	\begin{tabular}{|c|c|}
		\hline
		& Overdispersion test  \\ \hline \hline
		& $H_{0} : \; \psi = 1$ \\ 
		& $H_{1} : \; \psi > 1$ \\ 
		& $Z=4.3942$                  \\ 
		& $p-value = 5.558 \times 10^{-6}$    \\ \hline
	\end{tabular}
	\caption{Overdispersion test}\label{tab4}
\end{table}
With a threshold of $1\%$, we reject the null hypothesis, i.e a quasi-Poisson model, with the variance proportional to the mean, should be more reasonable.\\
After performing a quasi-Poisson regression, we get the same estimate of parameters as in table \ref{tab3}, but the overdispersion parameter is taken to be $\hat{\psi}=52601.93$.\\
Hence, we can predict the incremental claims payments as $\hat{Y}_{ij}$ (see \ref{tab5}).\\
\begin{table}[!h]
	\centering
	\begin{tabular}{|c|c|c|c|c|c|c|c|c|c|c|c}
		\cline{3-12}
	\multicolumn{2}{c|}{} &	\multicolumn{10}{|c|}{\centering Development Year}\\ \hline
		\multicolumn{2}{|c|}{$i/j$} &  1 &  2 &  3 &  4 &  5 &  6 & 7 & 8 & 9 & 10  \\ \hline
		\multirow{12}{6mm}{\rotatebox{90}{ Origin Year}} & 1  & 270061.4  & 672616.7 & 704494.1 & 753437.8 & 417350.2  & 292570.6 & 268343.5 & 182034.7 & 272606.0 & 67948.00  \\
		 & 2  & 376125.0 & 936779.4 &  981176.3 & 1049342.0 & 581259.8 & 407474.4 & 373732.4 & 253526.8 & 379669.0  & 94633.81  \\ 
		 & 3 & 372325.3 & 927315.9 & 971264.3 & 1038741.3 & 575387.7 & 403358.0 & 369956.9 & 250965.6  & 375833.5 & 93677.80  \\ 
		& 4 & 366724.0 & 913365.1 & 956652.3 & 1023114.2 & 566731.5 & 397289.8 & 364391.2 & 247190.0   & 370179.3 & 92268.49  \\ 
		 & 5 & 336287.3 & 837559.2 & 877253.8 & 938199.6 & 519694.9 & 364316.2 & 334148.1 & 226674.1 & 339455.9 & 84610.55  \\ 
		 & 6 & 353798.1 & 881171.9 & 922933.4 & 987052.7 & 546756.0 & 383286.6 & 351547.5 & 238477.3 & 357131.7 & 89016.32  \\ 
		 & 7 & 391841.7 & 975923.4 & 1022175.5 & 1093189.5 & 605548.1 & 424501.0 & 389349.1 & 264120.5 & 395533.7 & 98588.16  \\ 
		& 8 & 469647.5 & 1169707.2 & 1225143.3 & 1310258.2 & 725788.5 & 508791.9 & 466660.0 & 316565.5 & 474072.7 & 118164.27 \\ 
		& 9 & 390560.8 & 972733.2 & 1018834.1 & 1089616.0 & 603568.6 & 423113.4 & 388076.4 & 263257.2 & 394240.8 & 98265.88  \\ 
		& 10 & 344014.0 & 856803.5 & 897410.1 & 959756.3 & 531635.7 & 372687.0 & 341825.7 & 231882.4 & 347255.4 & 86554.62 \\ \hline
	\end{tabular}
	\caption{Prediction of $Y_{ij}$ from table \ref{tab2}}\label{tab5}
\end{table}
Then, the total amount of reserves is the sum of $\hat{Y}_{ij}$ for calendar year beyond 10, i.e
\begin{equation}
\hat{R}= \sum\limits_{i+j> 10}\hat{Y}_{ij}=18,680,856.
\end{equation}
To compute the standard error, mean squared error and the Reserve Prediction Error (Kaas et al.,
2008), Bootstrapping based on residuals of quasi-Poisson regression is needed (see \ref{tab6}). Recall that \\
\begin{table}[!h]
	\centering
	\begin{tabular}{|c|c|}
	\hline
	Reserve Prediction Error ($EP_{\hat{R}}$) & 2,882,413 \\ \hline
	Reserve Standard Deviation ($\sigma_{\hat{R}}$)	& 2,706,597 \\ \hline
	Mean squared error ($MSE_{\hat{R}}$) & 2,945,661 \\ \hline
	\end{tabular}
	\caption{Variability of $\hat{R}$}\label{tab6}
\end{table}\\
\section{A Hybrid Log-Poisson Regression for Loss Reserving}
In this section, we present a new model, which is the extension of the classical log-Poisson Regression developed by Mack (1991) in loss reserving.\par
Mack (1991) assumes that incremental payments $Y_{ij}$ are Poisson distributed, i.e,
\begin{equation}
Y_{ij} \sim \mathcal{P}(e^{\nu_{ij}}) \Rightarrow \mathbb{E}(Y_{ij})=e^{\nu_{ij}}\; \forall (i,j)\in \mathcal{T}_{k}.
\end{equation} 
Denote
\begin{equation}
\mathbb{E}(Y_{ij})=\lambda_{ij},\; \forall (i,j)\in \mathcal{T}_{k}.
\end{equation}
So the model becomes
\begin{equation}\label{73}
\ln (\lambda_{ij})=\tau + \alpha_{i} + \beta_{j},\;  \forall (i,j)\in \mathcal{T}_{k}.
\end{equation}
In the model, we consider a hybrid log-Poisson regression with minimum fuzziness and Maximum Likelihood (ML) criterion.\par
We assume that uncertainty about incremental payments in the run-off triangle is due both to fuzziness and randomness. Then the estimate $\hat{\lambda}_{ij}$ of $\lambda_{ij}$ will be obtain by the use of weighted function of the FN $\tilde{\lambda}^{Est}_{ij}$; where $\tilde{\lambda}^{Est}_{ij}$ denote the estimate of the FN $\tilde{\lambda}_{ij}.$\par
We suppose that $\tilde{Y}_{ij}$ is a fuzzy Poisson random variable (Buckley, 2006), i.e,
\begin{equation}
\tilde{Y}_{ij} \sim \tilde{\mathcal{P}}(\lambda_{ij}), \; \forall (i,j) \in \mathcal{T}_{k}.
\end{equation}  
According to (Buckley, 2006), the fuzzy expected value of $\tilde{Y}_{ij}$ is defined by its $h-$level, i.e,
\begin{align*}
\big[\mathbb{E}_{F}(\tilde{Y}_{ij})\big]_{h} & = \{\sum\limits_{x=0}^{+\infty}xe^{-\lambda_{ij}}\dfrac{(\lambda_{ij})^{x}}{x!}\mid \lambda_{ij} \in [\tilde{Y}_{ij}]_{h} \}\\
& =\{  \lambda_{ij}\mid \lambda_{ij} \in [\tilde{Y}_{ij}]_{h} \}\\
& = \tilde{\lambda}_{ij},
\end{align*}
where $\mathbb{E}_{F}(\cdot)$ is the fuzzy mean operator. So the fuzzy mean is just the fuzzification of the crisp mean.\par
Suppose that $\tilde{\lambda}_{ij}=(\lambda_{ij}^{L}, \lambda_{ij}^{c}, \lambda_{ij}^{R})$. Denote $\ln (\tilde{\lambda}_{ij})  = \tilde{\nu}_{ij}$.\\
Then our hybrid model built from the log-Poisson regression can be defined as follows
\begin{equation}\label{71}
\tilde{\nu}_{ij}=\tilde{\tau} + \tilde{\alpha}_{i} + \tilde{\beta}_{j}, \; \forall (i,j)\in  \mathcal{T}_{k}, \; \text{and}\; \tilde{\alpha}_{1}=\tilde{\beta}_{1}=0_{\mathbb{R}^{3}}
\end{equation}
where
\begin{align*}
\tilde{\nu}_{ij} & =(\nu^{L}_{ij}, \nu^{c}_{ij}, \nu^{R}_{ij}) \in \mathbb{R}^{3}\\
\tilde{\tau} & = (\tau^{L}, \tau^{c}, \tau^{R})  \in \mathbb{R}^{3}\\
\tilde{\alpha}_{i} & = (\alpha^{L}_{i}, \alpha^{c}_{i}, \alpha^{R}_{i})\in \mathbb{R}^{3}\\
\tilde{\beta}_{j} & =(\beta^{L}_{j}, \beta^{c}_{j}, \beta^{R}_{j}) \in \mathbb{R}^{3}
\end{align*}
are ATFN and\\
\begin{equation}
\left\{ \begin{array}{cl}
\nu^{L}_{ij} & = \tau^{L}+\alpha^{L}_{i}+\beta^{L}_{j}\\
\nu^{c}_{ij} & = \tau^{c}+\alpha^{c}_{i}+\beta^{c}_{j}\\
\nu^{R}_{ij} & = \tau^{R}+\alpha^{R}_{i}+\beta^{R}_{j}
\end{array}\right. 
\end{equation}
\subsection{Modeling Methodology Setup}
The steps for implementing the new model can be carried out as follows
\begin{description}
\item[\textbf{\textit{1. Modeling of the centres of fuzzy parameters.}}]
We apply the log-Poisson Regression to the incremental payments $Y_{ij}$ of Table \ref{tab1}, i.e the Ordinary Least Squared (OLS) regression on $\nu_{ij}$ and the model is
\begin{equation}\label{74}
\nu_{ij}=\tau + \alpha_{i} + \beta_{j} + \epsilon_{ij}
\end{equation}
where $\epsilon_{ij}$ are zero mean errors terms and are assumed to be uncorrelated.
\item[\textbf{\textit{2. Estimation and optimization of fuzzy parameters of the model.}}]
We consider the parameters of the model \eqref{73} as FN and the model become as in \eqref{74}. The centres of fuzzy coefficients of the model \eqref{71} are estimated in step \textbf{1.}\\
Let us assume that these estimates are $\hat{\tau}^{c}, \hat{\alpha}^{c}_{i}$ and $\hat{\beta}^{c}_{j}$.\par
Then
\begin{equation}
\hat{\nu}_{ij}^{c}=\hat{\tau}^{c} + \hat{\alpha}^{c}_{i} + \hat{\beta}^{c}_{j}.
\end{equation}
To estimate the parameters $\tau^{L}, \tau^{R}, \alpha_{i}^{L}, \alpha_{i}^{R}, \beta_{j}^{L}$ and $\beta_{j}^{R}, i \in \{2,3,\ldots , k\}; \; j \in \{2,3,\ldots, k-i \}$ in \eqref{71}, we present the problem as a linear programming problem. This linear programming problem can be written as following :
\begin{equation}
\left\{ \begin{array}{cl}
\min : & z=  \sum\limits_{i=1}^{k}\sum\limits_{j=1}^{k-i}(\nu_{ij}^{R} - \nu_{ij}^{L})\\
& \\
\text{s/t} & \left\{ \begin{array}{cl}
& h^{\ast}\cdot\hat{\nu}_{ij}^{c} - (1-h^{\ast})\cdot\nu_{ij}^{L} \leqslant \ln (Y_{ij})\\
& \\
& h^{\ast}\cdot\hat{\nu}_{ij}^{c} + (1-h^{\ast})\cdot\nu_{ij}^{R} \geqslant \ln (Y_{ij})\\
&\\
& \nu_{ij}^{L}, \nu_{ij}^{R} \in \mathbb{R},\; i=1,2,\ldots, k; j=1,2,\ldots, k-i
\end{array}\right.
\end{array}\right.
\end{equation}\label{90}
where $h^{\ast}$ is the optimized value of $h$ according to Chen et al. (2016) and $\alpha_{1}^{L}=\alpha_{1}^{c}=\alpha_{1}^{R}=0$ and $\beta_{1}^{L}=\beta_{1}^{c}=\beta_{1}^{R}=0$.
\item[\textbf{\textit{3. Defuzzification of fuzzy parameters and Prediction of incremental payments.}}]
At this step, we predict the incremental claims payments of the non fill part of the run-off triangle \ref{tab1} using model \eqref{71} estimated at the step \textbf{2.} Denote$\tilde{Y}_{ij}^{Est}$ that predicted claim payment. Then we have $\tilde{Y}_{ij}^{Est}=e^{\tilde{\nu}_{ij}} $ for $i=1, \ldots, k$ and $j\geqslant k-i+1$.\par
The $h-$level of $\tilde{\nu}_{ij}$ is defined as following :
\begin{align}
& \tilde{\nu}_{ij}(h)=\big[ h\cdot\nu_{ij}^{c} - (1-h)\cdot\nu_{ij}^{L}, h\nu_{ij}^{c} + (1-h)\cdot\nu_{ij}^{R} \big]\\
& \tilde{Y}_{ij}^{Est}(h)= \big[e^{h\cdot\nu_{ij}^{c} - (1-h)\cdot\nu_{ij}^{L}}, e^{h\cdot\nu_{ij}^{c} + (1-h)\cdot\nu_{ij}^{R}}\big].
\end{align}
We use now the concept of of weighted function (Definition \ref{def1}) applied on $\tilde{Y}_{ij}^{Est}$ in order to compute the crisp value $Y_{ij}^{Est}$ of $\tilde{Y}_{ij}^{Est}$ rather than the expected value of FN (de Campos Ib\'{a}\~{n}ez
and Mu\~{n}oz, 1989) used in de Andr\'{e}s S\'{a}nchez (2006) which have to take into account an arbitrary parameter $\beta \in [0,1]$ which is the decision-maker risk aversion.
\begin{equation}
Y_{ij}^{Est} = \dfrac{\frac{1}{2}\bigg[ \int_{0}^{1}e^{h\cdot\nu_{ij}^{c} - (1-h)\cdot\nu_{ij}^{L}}hdh + \int_{0}^{1}e^{h\cdot\nu_{ij}^{c} + (1-h)\cdot\nu_{ij}^{R}}hdh \bigg]}{\int_{0}^{1}hdh}
\end{equation}
But $2\int_{0}^{1} hdh = [h^2]_{0}^{1} =1$.\\
and
\begin{align}
\int_{0}^{1}e^{h\cdot\nu_{ij}^{c} - (1-h)\cdot\nu_{ij}^{L}}hdh & = \bigg[ \dfrac{h}{\nu_{ij}^{c} + \nu_{ij}^{L}}e^{-\nu_{ij}^{L} + h(\nu_{ij}^{c} + \nu_{ij}^{L})}\bigg]_{0}^{1} - \dfrac{1}{\nu_{ij}^{c} + \nu_{ij}^{L}}\int_{0}^{1}e^{-\nu_{ij}^{L} + h(\nu_{ij}^{c} + \nu_{ij}^{L})}dh; \;  \nu_{ij}^{c} + \nu_{ij}^{L} \neq 0\\
& = \dfrac{e^{\nu_{ij}^{c}}}{\nu_{ij}^{c} + \nu_{ij}^{L}} - \bigg[\dfrac{e^{-\nu_{ij}^{L} + h(\nu_{ij}^{c} + \nu_{ij}^{L})}}{(\nu_{ij}^{c} + \nu_{ij}^{L})^2} \bigg]_{0}^{1};\; \nu_{ij}^{L} + \nu_{ij}^{c} \neq 0\\
& = \dfrac{e^{\nu_{ij}^{c}}}{\nu_{ij}^{c} + \nu_{ij}^{L}} - \bigg(\dfrac{e^{\nu_{ij}^{c}}}{(\nu_{ij}^{c} + \nu_{ij}^{L})^2} - \dfrac{e^{ - \nu_{ij}^{L}}}{(\nu_{ij}^{c} + \nu_{ij}^{L})^2} \bigg);\; \nu_{ij}^{c} + \nu_{ij}^{L} \neq 0\\
& = \dfrac{(\nu_{ij}^{c} + \nu_{ij}^{L})e^{\nu_{ij}^{c}} - e^{\nu_{ij}^{c}} + e^{-\nu_{ij}^{L}}}{(\nu_{ij}^{c} + \nu_{ij}^{L})^2},\; \nu_{ij}^{c} + \nu_{ij}^{L} \neq 0.
\end{align}
Likewise
\begin{equation}
\int_{0}^{1}e^{h\nu_{ij}^{c} + (1-h)\nu_{ij}^{R}}hdh=\dfrac{e^{\nu_{ij}^{c}}(\nu_{ij}^{c} - \nu_{ij}^{R} - 1) + e^{\nu_{ij}^{R}}}{(\nu_{ij}^{c} - \nu_{ij}^{R})^2},\; \nu_{ij}^{c} - \nu_{ij}^{R} \neq 0.
\end{equation}
Hence
\begin{equation}
Y_{ij}^{Est} =  \dfrac{(\nu_{ij}^{c} + \nu_{ij}^{L} - 1)e^{\nu_{ij}^{c}} + e^{-\nu_{ij}^{L}}}{(\nu_{ij}^{c} + \nu_{ij}^{L})^2} + \dfrac{e^{\nu_{ij}^{c}}(\nu_{ij}^{c} - \nu_{ij}^{R} - 1) + e^{\nu_{ij}^{R}}}{(\nu_{ij}^{c} - \nu_{ij}^{R})^2},\; \nu_{ij}^{c} + \nu_{ij}^{L} \neq 0\; \nu_{ij}^{c} - \nu_{ij}^{R} \neq 0.
\end{equation}\label{77}
\item[\textbf{\textit{4. Estimation of the outstanding reserve.}}]
The loss reserve corresponding to the $i^{th}$ accident year is 
\begin{equation}
R_{i}^{Est}=\sum\limits_{j=k-i+1}^{k}Y_{ij}^{Est}
\end{equation}
and the outstanding reserve is 
\begin{equation}
R^{Est}=\sum\limits_{i=1}^{k}\sum\limits_{j=k-i+1}^{k}Y_{ij}^{Est}
\end{equation}
\end{description}
\subsection{Implementation of the Method on Data}
We apply our model \eqref{71} on the run-off triangle from Taylor and Ashe (1983) (see \ref{tab2}). We first fit the log-Poisson regression on this data to estimate the centres ($\tau^{c}, \alpha_{i}^{c}, \beta_{j}^{c}$) of the fuzzy parameters ($\tilde{\tau}, \tilde{\alpha}_{i}, \tilde{\beta}_{j}$). The results is presented in Table \ref{tab3}. Now we use the linear programming problem in \eqref{90} to estimate the left and right spreads ($\tau^{L},\alpha_{i}^{L}, \beta_{j}^{L}, \tau^{R},\alpha_{i}^{R}, \beta_{j}^{R}$). Using Tables \ref{tab2}-\ref{tab3} and equation \eqref{90}, we estimate the fuzzy parameters of the new model with \texttt{Matlab.} After we use the concept of weighted function to compute the crisp value $Y_{ij}^{Est}$(equation \eqref{77}). The results displayed in Table \ref{tab8} and \ref{tab9}.\\
\begin{table}[!h]
	\centering
	\begin{tabular}{|c|c|}
		\hline
		 $(\tilde{\alpha}_{i}^{Est})_{2 \leqslant i \leqslant 10}$ & $(\tilde{\beta}_{j}^{Est})_{2 \leqslant j \leqslant 10}$ \\ \hline \hline
		 $(0.3310,0.3312, 0.3313)$ & $(0.9125 ,0.9125, 0.9125)$ \\
		 $(0.3211,0.3211, 0.3211)$ & $(0.9588 ,0.9588, 0.9588)$ \\
		 $(0.3059,0.3059,0.3060)$ & $(1.0259,1.0259,1.1864)$ \\
		 $(0.2193,0.2193,0.2193)$ & $(0.4352,0.4352,0.4353)$\\
		 $(0.2700,0.2700,0.2701)$ & $(0.0800,0.0800,0.5243)$ \\
		 $(0.3722,0.3722,0.3722)$ & $(-0.0063,-0.0063,0.0656)$ \\ 
		 $(0.5533,0.5533,0.5533)$ & $(-0.3945,-0.3944,-0.3944)$ \\ 
		 $(0.3689,0.3689,0.3689)$ & $(0.0093,0.0093,0.0094)$ \\ 
		 $(0.240,0.240,0.2420)$ & $(-1.3799,-1.3799,-0.0760)$ \\ \hline
		\multicolumn{2}{|c|}{ $ \tilde{\tau}^{Est} = (12.506,12.506,12.8244);\quad h^{\star} = 0.115$} \\ \hline
	\end{tabular}
	\caption{Estimation of parameters of our models}\label{tab8}
\end{table}
\begin{table}[!h]
	\centering
	\begin{tabular}{|c|c|c|c|c|c|c|c|c|c|c|c|}
		\cline{3-12}
		\multicolumn{2}{c|}{} &	\multicolumn{10}{|c|}{\centering Development Year}\\ \hline
		\multicolumn{2}{|c|}{$i/j$} &  1 &  2 &  3 &  4 &  5 &  6 & 7 & 8 & 9 & 10 \\ \hline
		\multirow{12}{6mm}{\rotatebox{90}{ Origin Year}} & 1  & 160879.889 & 398991.237 & 417815.351 & 471378.428 & 248094.604 & 202870.489 & 163724.560 & 108664.307 & 162380.941 & 65947.415 \\ 
		& 2 & 223697.783 & 554890.532 & 581075.012 & 655624.053 & 345000.220 & 282153.114 & 227661.516 & 151079.328 & 225785.445 & 91732.722 \\ 
		& 3 & 222198.647 & 549318.932 & 575240.344 & 649038.546 & 341537.092 & 279318.881 & 225376.552 & 149563.648 & 223519.632 & 90810.394 \\ 
		& 4 & 218135.115 & 541084.422 & 566617.042 & 639307.149 & 336418.736 & 275131.901 & 212603.384 & 147323.497 & 220170.827 & 89449.093 \\ 
		& 5 & 200115.161 & 496361.251 & 519782.296 & 586450.621 & 308619.792 & 252387.076 & 203658.495 & 135156.565 & 201982.590 & 82051.110 \\ 
		& 6 & 210478.713 & 522082.044 & 546717.470 & 616849.114 & 324607.333 & 265468.126 & 214206.815 & 142153.980 & 212442.923 & 86306.013 \\ 
		& 7 & 233007.431 & 577996.713 & 605272.178 & 682931.930 & 359362.252 & 293903.243 & 237137.054 & 157365.016 & 212419.327 & 95553.951 \\ 
		& 8 & 279034.306 & 692239.443 & 724909.289 & 817955.617 & 430369.905 & 352003.602 & 283985.052 & 188440.664 & 281638.791 & 114451.989 \\ 
		& 9 & 232243.357 & 576100.293 & 603286.215 & 680690.598 & 358183.506 & 292938.796 & 236359.356 & 156849.131 & 234410.827 & 95240.262 \\ 
		& 10 & 204415.080 & 507034.083 & 530959.088 & 599076.133 & 315253.479 & 257828.539 & 208036.859 & 138059.710 & 206322.671 & 83830.125 \\ \hline
	\end{tabular}
	\caption{Prediction of $Y_{ij} (Y_{ij}^{Est})$ from our model}\label{tab9}
\end{table}	
From table \ref{tab8}, the parameter $\tau$ of the original model (\eqref{62}) is the most affected by the fuzziness. Some other parameters remains crisp, i.e after estimation of fuzzy parameters, those parameters are not affected by fuzziness. In that case, the centre, left and right spreads are equal.\\
Table \ref{tab9} have been computed using the estimations of fuzzy parameters of our model and the weighted function of FN.\par
\textbf{Remark :} The optimized values of fuzzy parameters of model \eqref{73} could be ATFN or STFN. That why in the linear programming problem \eqref{90} we did not specified a restriction on the fuzzy parameters. \par
From Table \ref{tab9}, we compute the outstanding reserve, i.e
\begin{equation}
R^{Est}=\sum\limits_{i=1}^{10}\sum\limits_{j=10-i+1}^{10}Y_{ij}^{Est}=16,735,378.64
\end{equation}	
The Mean Square Error, the Standard Deviation and the Error reserve prediction of the outstanding reserve could be computed using the bootstrap method, i.e 
\begin{enumerate}
\item[\textbf{(a)}]
We calculate the outcomes of the Pearson residuals from $Y_{ij}$ of table \ref{tab2} and $Y_{ij}^{Est}$ of our model.
\item[\textbf{(b)}]
We resample from the adjusted residuals, with replacement.
\item[\textbf{(c)}]
We use this set of residuals and the estimated values $Y_{ij}^{Est}$ to create a new suitable pseudo-history.
\item[\textbf{(d)}]
We compute the fitted values, and use the sum of the future part as an estimate of the reserve to be held.
\item[\textbf{(e)}]
At the end of the loop, store the simulated total payments and the estimated reserve to be held.
\end{enumerate}
The results are displayed in table \ref{tab10}.\\
\begin{table}[!h]
	\centering
	\begin{tabular}{|c|c|c|}
		\hline
		Reserve Prediction Error ($EP_{R^{Est}}$) & 2,808,982 \\ \hline
		Reserve Standard Deviation ($\sigma_{R^{Est}}$)	& 2,081,045.17969647 \\ \hline
		Mean squared error ($MSE_{R^{Est}}$) & 3,721,611\\ \hline
	\end{tabular}
	\caption{Variability of $R^{Est}$}\label{tab10}
\end{table}\\
According to Reserve Prediction Error and Reserve Standard Deviation criteria (see table \ref{tab10}), our hybrid model produce good results than the classical log-Poisson Regression since we get a Reserve Prediction Error and a Reserve Standard Deviation smaller that those computed from the log-Poisson Regression. But the Reserves Mean Squared error in our model is greater than the the one in the classical log-Poisson Regression. That means, the estimator of Reserve in the Hybrid Model is more biased that the one from the classical model.

\section{Conclusion} 
This paper has considered the relevance of Hybrid Models in loss reserving framework, mainly when we are in presence of vague information like in medical insurance (Straub and Swiss, 1988). Those models could give best result compared to stochastic models. In this article, the Hybrid Model that has been suggested is in GLM framework (eg Log-Poisson Regression). The optimize $h-$ value (Chen
et al., 2016) has been taken into account in the linear programming problem. Bootstrap procedure has been used to compute the Reserve Prediction Error, the Reserve Standard Deviation and the Mean Square Error. From a numerical application, it has been shown that the model produce good results of reserve than the classical log-Poisson regression according to Reserve Prediction Error and Standard Deviation.\par     
In brief, the Hybrid Model developed in this paper produce best results than the classical log-Poisson regression according to Error prediction and Standard Deviation criteria but the mean square error criterion is not satisfied. This is a weakness of our model. The hard computation aspect of our hybrid model constitute another weakness compare to the simplest computation of a classical log-Poisson regression. 
\section*{Acknowledgements}
This work was supported by African Union through Pan African University.\newpage

\textbf{{\large References}}\\

Apaydin, A. and Baser, F. (2010). Hybrid fuzzy least-squares regression analysis in claims reserving with geometric separation method. Insurance: Mathematics and Economics, 47(2):113-122.\par
Asai, H. T.-S. U.-K. (1982). Linear regression analysis with fuzzy model. IEEE Trans. Systems Man Cybern, 12:903-907.\par
Baser, F. and Apaydin, A. (2010). Calculating insurance claim reserves with hybrid fuzzy least squares regression analysis. Gazi University Journal of Science, 23(2):163-170.\par
Bornhuetter, R. L. and Ferguson, R. E. (1972). The actuary and ibnr. In Proceedings of the casualty actuarial society, volume 59, pages 181-195.\par
Buckley, J. J. (2006). Fuzzy probability and statistics, volume 196. Springer Science \& Business Media.\par
Chang, Y.-H. O. (2001). Hybrid fuzzy least-squares regression analysis and its reliability measures. Fuzzy Sets and Systems, 119(2):225-246.\par
Chen, F., Chen, Y., Zhou, J., and Liu, Y. (2016). Optimizing h value for fuzzy linear regression with asymmetric triangular fuzzy coefficients. Engineering Applications of Artificial Intelligence, 47:16-24.\par
de Andr\'{e}s S\'{a}nchez, J. (2006). Calculating insurance claim reserves with fuzzy regression. Fuzzy sets and systems, 157(23):3091-3108.\par
de Andr\'{e}s S\'{a}nchez, J. (2007). Claim reserving with fuzzy regression and taylors geometric separation method. Insurance: Mathematics and Economics, 40(1):145-163.\par
de Andr\'{e}s S\'{a}nchez, J. (2012). Claim reserving with fuzzy regression and the two ways of anova. Applied Soft Computing, 12(8):2435-2441.\par
de Andr\'{e}s S\'{a}nchez, J. (2014). Fuzzy claim reserving in non-life insurance. Comput. Sci. Inf. Syst., 11(2):825-838.\par
de Campos Ib\'{a}\~{n}ez, L. M. and Mu\~{n}oz, A. G. (1989). A subjective approach for ranking fuzzy
numbers. Fuzzy sets and systems, 29(2):145-153.\par
Dubois, D. and Prade, H. (1978). Operations on fuzzy numbers. International Journal of systems science, 9(6):613-626.\par
Dubois, D. and Prade, H. (1988). Fuzzy numbers : An overview. In Analysis of Fuzzy Information, pages 3-39. J. C Bezdek.\par
England, P. D. and Verrall, R. J. (2002). Stochastic claims reserving in general insurance. British Actuarial Journal, 8(03):443-518.\par
Ishibuchi, H. and Nii, M. (2001). Fuzzy regression using asymmetric fuzzy coefficients and fuzzified neural networks. Fuzzy Sets and Systems, 119(2):273-290.\par
Kaas, R., Goovaerts, M., Dhaene, J., and Denuit, M. (2008). Modern actuarial risk theory: using R, volume 128. Springer Science \& Business Media.\par
Kauffman, A. and Gupta, M. M. (1991). Introduction to fuzzy arithmetic, theory and application.\par
Kaufmann, A. and Gupta, M. M. (1991). Introduction to Fuzzy Arithmetic. Van Nostrand Reinhold Company.\par
Lai, Y.-J. and Hwang, C.-L. (1992). Fuzzy mathematical programming. In Fuzzy Mathematical
Programming, pages 74-186. Springer.\par
Linnemann, P. (1984). van eeghen j. (1981): Loss reserving methods. surveys of actuarial studies no. 1. nationale-nederlanden n.v., rotterdam. 114 pages. ASTIN Bulletin: The Journal of the
International Actuarial Association, 14(01):87-88.\par
Mack, T. (1991). A simple parametric model for rating automobile insurance or estimating ibnr claims reserves. Astin bulletin, 21(01):93-109.\par
Moskowitz, H. and Kim, K. (1993). On assessing the h value in fuzzy linear regression. Fuzzy sets and systems, 58(3):303-327.\par
Renshaw, A. E. and Verrall, R. J. (1998). A stochastic model underlying the chain-ladder technique. British Actuarial Journal, 4(04):903-923.\par
Straub, E. and Swiss, A. A. (1988). Non-life insurance mathematics. Springer.
Taylor, G. (1986). Claims reserving in non-life insurance. Insurance series. North-Holland.\par
Taylor, G., McGuire, G., and Greenfield, A. (2003). Loss reserving: past, present and future.
Taylor, G. C. and Ashe, F. (1983). Second moments of estimates of outstanding claims. Journal of Econometrics, 23(1):37-61.\par
W\"{u}thrich, M. V. and Merz, M. (2008). Stochastic claims reserving methods in insurance, volume 435. John Wiley \& Sons.\par
Yager, R. R. and Filev, D. (1999). On ranking fuzzy numbers using valuations. International Journal of Intelligent Systems, 14(12):1249-1268.\par
Zadeh, L. A. (1965). Fuzzy sets. Information and control, 8(3):338-353.\par
Zadeh, L. A. (1975b). The concept of a linguistic variable and its application to approximate reasoningi. Information sciences, 8(3):199-249.\par
Zadeh, L. A. (1975c). The concept of a linguistic variable and its application to approximate reasoningii. Information sciences, 8(4):301-357.\par
Zadeh, L. A. (1975d). The concept of a linguistic variable and its application to approximate reasoning-iii. Information sciences, 9(1):43-80.\par
Zadeh, L. A. and al. (1975a). Calculus of fuzzy restrictions. Electronics Research Laboratory, University of California.
\end{document}